\pdfoutput=1
\documentclass[letterpaper,11pt]{article}

\usepackage{xcolor}
\usepackage[colorlinks = true]{hyperref}
\hypersetup{
  pdftitle = {Trading inverses for an irrep in the Solovay-Kitaev theorem},
  pdfauthor = {Adam Bouland and Maris Ozols}
}
\definecolor{darkred}  {rgb}{0.5,0,0}
\definecolor{darkblue} {rgb}{0,0,0.5}
\definecolor{darkgreen}{rgb}{0,0.5,0}
\hypersetup{
  urlcolor   = blue,         
  linkcolor  = darkblue,     
  citecolor  = darkgreen,    
  filecolor  = darkred       
}

\usepackage{amssymb,amsmath,amsthm}
\usepackage{thmtools,thm-restate}

\usepackage{fullpage}

\usepackage{mathpazo}

\newcommand{\Thm}[1]{\hyperref[thm:#1]{Theorem~\ref*{thm:#1}}}
\newcommand{\Clm}[1]{\hyperref[claim:#1]{Claim~\ref*{claim:#1}}}
\newcommand{\Sec}[1]{\hyperref[sec:#1]{Section~\ref*{sec:#1}}}
\newcommand{\App}[1]{\hyperref[app:#1]{Appendix~\ref*{app:#1}}}
\newcommand{\Eq}[1]{eq.~\hyperref[eq:#1]{(\ref*{eq:#1})}}

\usepackage{mathtools}
\mathtoolsset{centercolon}
\DeclarePairedDelimiter{\set}{\lbrace}{\rbrace}
\DeclarePairedDelimiter{\abs}{\lvert}{\rvert}
\DeclarePairedDelimiter{\norm}{\lVert}{\rVert}
\DeclarePairedDelimiter{\of}{\lparen}{\rparen}

\newcommand{\ct}{^{\dagger}}

\newcommand{\U}[1]{\mathrm{U}(#1)}
\newcommand{\SU}[1]{\mathrm{SU}(#1)}
\newcommand{\SL}[1]{\mathrm{SL}(#1,\C)}
\newcommand{\PU}[1]{\mathrm{PU}(#1)}
\newcommand{\PSU}[1]{\mathrm{PSU}(#1)}
\renewcommand{\S}[1]{\mathrm{S}_{#1}} 
\newcommand{\C}{\mathbb{C}}
\newcommand{\R}{\mathbb{R}}
\newcommand{\Z}{\mathbb{Z}}

\DeclareMathOperator{\Tr}{Tr}
\DeclareMathOperator{\sgn}{sgn} 
\DeclareMathOperator{\poly}{poly}
\DeclareMathOperator{\polylog}{polylog}

\newcommand{\idp}{\varepsilon} 
\newcommand{\G}{\mathcal{G}} 
\newcommand{\Op}{\mathcal{O}}
\renewcommand{\rho}{R} 

\newtheorem{theorem}{Theorem}
\newtheorem*{theorem*}{Theorem}

\newtheorem*{conjecture*}{Conjecture}

\newtheorem{claim}{Claim}

\newcommand{\partitle}[1]{\textit{#1.}}

\usepackage[affil-it]{authblk}
\title{Trading inverses for an irrep in the Solovay-Kitaev theorem}
\author[1]{Adam Bouland}
\author[2]{M\={a}ris Ozols}
\affil[1]{Department of EECS, UC Berkeley, Berkeley, CA USA}
\affil[2]{QuSoft and University of Amsterdam, Amsterdam, Netherlands}

\begin{document}
\date{}
\maketitle

\begin{abstract}
The Solovay-Kitaev theorem states that universal quantum gate sets can be exchanged with low overhead.
More specifically, any gate on a fixed number of qudits can be simulated with error $\epsilon$ using merely $\polylog(1/\epsilon)$ gates from any finite universal quantum gate set $\G$.
One drawback to the theorem is that it requires the gate set $\G$ to be closed under inversion.
Here we show that this restriction can be traded for the assumption that $\G$ contains an irreducible representation of any finite group $G$.
This extends recent work of Sardharwalla \emph{et al.}~\cite{refocusing},
and applies also to gates from the special linear group.
Our work can be seen as partial progress towards the long-standing open problem of proving an inverse-free Solovay-Kitaev theorem \cite{DawsonSolovayKitaev,kuperberg2009hard}.
\end{abstract}

\section{Introduction} \label{sec:Intro}

Quantum computing promises to solve certain problems exponentially faster than classical computers. For instance, quantum computers can factor integers \cite{shor1999polynomial}, simulate quantum mechanics \cite{berry2015hamiltonian}, and compute certain knot invariants \cite{aharonov2009polynomial} exponentially faster than the best known classical algorithms.
The power of quantum computing is formalized using the notion of quantum circuits, in which polynomial number of quantum gates are applied to a standard input state, and the answer to the computational problem is obtained by measuring the final state. This results in the complexity class $\textsf{BQP}$ (see \cite{nielsen2002quantum,Kitaev-book} for an introduction).

Each gate in the circuit is a unitary transformation drawn from some finite \emph{gate set} $\G$; it represents elementary quantum operations that can be performed in hardware and which may vary between different realizations of quantum computing. Each gate can act on at most some finite number $k$ of quantum systems at a time, where each individual system (or \emph{qudit}) has $d$ levels. A gate set $\G$ is called \emph{universal}\footnote{This is also known as \emph{physical universality}.} if it is capable of approximately generating any quantum transformation on a sufficiently large number of qudits~\cite{CLMO11}.

In general, the computational power of a quantum device may depend on the gate set $\G$ at its disposal. Clearly, if a gate set is not universal, it may have restricted computational power.\footnote{But not always! Some gate sets which are not physically universal are nevertheless capable of universal quantum computing via an encoding; this is known as \emph{encoded universality} \cite{nielsen2002quantum}.} But \emph{a priori}, the computational power of different universal gate sets could vary as well.
This is because universality simply implies that the gates from $\G$ densely generate all unitaries, but it does not specify how quickly one can approximate arbitrary gates.\footnote{By a simple counting argument, generic unitaries on an $n$-qubit system require $\tilde{\Omega}(2^n)$ gates to implement (even approximately) irrespectively of the gate set~\cite[Section~4.5.4]{nielsen2002quantum}.}

While $\mathsf{BQP}$ consists of those computations that use $\poly(n)$ gates on an $n$-bit input, the degree of the polynomial for a specific algorithm could in principle depend on the actual gate set used. For example, if we are given an $O(n)$-gate algorithm over some gate set and we want to implement it using another gate set $\G$, we have to compile each gate to accuracy $O(1/n)$ in terms of $\G$. However, if our compiler uses, say, $O(1/\epsilon)$ gates to achieve accuracy $O(\epsilon)$, the total number of gates would become $O(n^2)$. This would be a strange situation for quantum computation, since the runtime of polynomial-time algorithms would be defined only up to polynomial factors. In particular, this would render Grover's speed-up useless.

Fortunately, this is not the case since the \emph{Solovay-Kitaev theorem} \cite{Kitaev,Kitaev-book,nielsen2002quantum,DawsonSolovayKitaev} (see also \cite{ChildsQA,MarisSK}) provides a better compiler, so long as the universal gate set $\G$ is closed under inversion.
More specifically, this theorem states that any universal gate set $\G$ can be used to simulate any gate $U$ from any other universal gate set to accuracy $\epsilon$ using only $\polylog(1/\epsilon)$ gates from $\G$. Furthermore, there is an efficient algorithm, the \emph{Solovay-Kitaev algorithm}, to perform this conversion between the gate sets.

Before formally stating the Solovay-Kitaev theorem, let us make a few remarks.
First, we can assume without loss of generality that all gates in $\G$ are single-qudit gates in some fixed dimension $d$. Indeed, if $\G$ contains multi-qudit gates or if $\G$ becomes universal only on some larger number of qudits, we can simply set the new dimension to be $d^k$ (for a sufficiently large constant $k$) and replace $\G$ by a larger gate set that consists of the original gates acting on all ordered subsets of $k$ systems.
Second, as we are now dealing with a single system, we can replace the universality of $\G$ by a requirement that $\G$ generates a \emph{dense} subgroup of $\SU{d}$~\cite{SK17}.
Third, we can assume that $\G$ is itself a subset of the \emph{special} unitary group $\SU{d}$ rather than $\U{d}$, since the global phase of a quantum gate has no physical effect. In fact, $\U{1}$ actually does \emph{not} satisfy the Solovay-Kitaev theorem, hence the theorem does not hold for $\U{d}$ either, because in general we cannot approximate the elements of $\U{d}$ accurately enough due to their global phase.

With this fine print aside, we are now ready to state the theorem.

\begin{theorem*}[Solovay-Kitaev theorem~\cite{DawsonSolovayKitaev}]
For any fixed $d \geq 2$, if $\G \subset \SU{d}$ is a finite gate set which is \textbf{closed under inverses} and densely generates $\SU{d}$, then there is an algorithm which outputs an $\epsilon$-approximation to any $U \in \SU{d}$ using merely
$O\of[\big]{\log^{3.97}(1/\epsilon)}$
elements from $\G$.
\end{theorem*}

Therefore if one wishes to change the gate set used for a $\mathsf{BQP}$ computation (which requires compiling each gate to $1/\poly$ accuracy), a change of gate set only incurs polylogarithmic overhead in the input size $n$. In particular this implies the runtimes of quantum algorithms based on inverse-closed gate sets are well-defined up to polylog factors in $n$; an $O(n^c)$ algorithm using one particular universal gate set implies an $\tilde{O}(n^c)$ algorithm using any other (inversion-closed) universal gate set.
It also implies that the choice of a particular universal gate set is unimportant for quantum computation; changing between gate sets incurs low overhead.

Given the central importance of the Solovay-Kitaev theorem to quantum computing, prior works have improved the theorem in various directions. For instance, a number of works (see, e.g., \cite{kliuchnikov2013synthesis,paetznick2014repeat,selinger2015efficient,msr2015fallback,bocharov2015efficient,sarnak2015letter,ross2014optimal,kliuchnikov2016practical,parzanchevski2017super}) have decreased the overheads of the Solovay-Kitaev theorem for particular inverse-closed gate sets by improving the exponent in the logarithm from 3.97 to 1 (which is optimal) or even by improving the hidden constants in front of the logarithm.
Such works are important steps towards making compilation algorithms practically efficient. Additionally prior work has shown a version of the Solovay-Kitaev algorithm for inverse-closed non-unitary matrices \cite{aharonov2007polynomial} and as well more general Lie groups \cite{kuperberg2009hard}. Note that there is also an information-theoretic non-algorithmic version of the Solovay-Kitaev theorem with exponent 1 for generic inverse-closed gate sets \cite{harrowrechtchuang}. This has subsequently been extended also to inverse-free gate sets \cite{varju2013random}.

In this work, rather than improving the overheads of the Solovay-Kitaev theorem, we work towards removing the assumption that the gate set contains inverses of all its gates. This is important for several reasons. First, on a theoretical level it would be surprising if the power of noiseless quantum computers could be gate set dependent.
Of course, in the real world one could apply fault-tolerance \cite{aharonov1997fault} to allow the use of approximate inverses in place of exact inverses, but it seems strange to have to resort to such a powerful technique to deal with a seemingly minor issue which is easily stated in a noiseless setting.
Furthermore, this would not answer the original mathematical question about how fast unitary gate sequences fill the space of all unitaries, since a fault-tolerant implementation corresponds to a completely positive rather than a unitary map (it implements the desired map on an encoded subspace of a larger-dimensional Hilbert space).

Second, an inverse-free Solovay-Kitaev theorem would be very helpful towards classifying the computational power of quantum gate sets.
It remains open\footnote{Even for the case of two-qubit gate sets~\cite{CLMO11,boulandCCC2016}!} to prove a classification theorem describing which gate sets $\G$ are capable of universal quantum computing, which are efficiently classically simulable, and which can solve difficult sampling problems like \textsc{BosonSampling} or \textsf{IQP} \cite{aaronsonboson,BJS2010}.
A number of recent works have made partial progress on this problem \cite{boulandCCC2016,cccs}.
However, a common bottleneck in these proofs is that they need to invoke the Solovay-Kitaev theorem on various ``postselection gadgets'' to argue that one can perform hard sampling problems, and the set of these gadgets is not necessarily closed under inversion.
In the above works this problem is tackled on an ad-hoc, case specific basis.
An inverse-free Solovay-Kitaev theorem would simplify these proofs and expand the frontier for gate classification.

Finally, such theorem would enable further progress in quantum Hamiltonian complexity where universal gate sets are used to encode computational instructions by the interaction terms of local Hamiltonians. The ground states of such Hamiltonians have very complicated structure and computing their ground energy is typically $\mathsf{QMA_{EXP}}$-complete \cite{GI13}, a phenomenon that can occur even when the local dimension of each individual subsystem is relatively small \cite{BCO16,BP17}. Since low local dimension is physically more relevant, it is desirable to minimize the dimensions of these constructions even further. A significant roadblock in this is the size of the universal gate set used to encode the computation. Since each gate contributes additional dimensions, one would like to have as few gates as possible. Considering how intricate and hard to optimize the known constructions \cite{BCO16,BP17} are, getting rid of inverses would be an easy way forward.

For the reasons outlined above, we believe this longstanding open problem (noted in \cite{DawsonSolovayKitaev,kuperberg2009hard}) is an important one to resolve.
In this work, we make partial progress towards this goal by replacing the inverse-closedness assumption with the requirement that the gate set contains any (projective) irreducible representation (a.k.a.\ \emph{irrep}) of a finite group.
Roughly speaking, a \emph{projective irrep} is a set of unitary matrices that form a group (up to a global phase) and that do not leave any non-trivial subspace invariant. A canonical example is the set of \emph{Pauli matrices} $\set{I,X,Y,Z}$.

\begin{restatable}[Solovay-Kitaev theorem with an irrep instead of inverses]{theorem}{ThmMain}\label{thm:ours}
For any fixed $d\geq2$, suppose $\G \subset \SU{d}$ is a finite gate set which densely generates $\SU{d}$, and furthermore $\G$ contains a (projective) irrep of some finite group $G$. Then there is an algorithm which outputs an $\epsilon$-approximation to any $U \in \SU{d}$ using merely $O(\polylog(1/\epsilon))$ elements from $\G$.
\end{restatable}

In other words, the inverses of some of the gates of $\G$---namely those which constitute an irrep of $G$---are also in $\G$, but the inverses of the remaining gates may not be in $\G$.
So we are trading inverses for some other structure in the gate set $\G$.
This extends recent work of Sardharwalla, Cubitt, Harrow and Linden~\cite{refocusing} which proved this theorem in the special case that $G$ is the Weyl (or generalized Pauli)\footnote{Note the Weyl operators only form a group up to global phase, but as we only require a \emph{projective irrep} they meet the criteria of our theorem.} group.
Sardharwalla \emph{et al.}'s result has already found application in gate set classification \cite{cccs}. We therefore expect that our result will likewise enable further progress on the gate set classification problem.
We also extend our theorem to the non-unitary case (see \Thm{oursSL} in \App{SLnC}), thus generalizing the (inverse-closed) non-unitary Solovay-Kitaev theorem of \cite{aharonov2007polynomial} (this may further extend to more general Lie groups as well following \cite{kuperberg2009hard}).
We expect that this version of the theorem will be particularly useful in gate classification as postselection gadgets are often non-unitary \cite{boulandCCC2016}.

\subsection{Proof techniques}
\label{sec:prooftechniques}

Our proof works in a similar manner to those in \cite{DawsonSolovayKitaev,refocusing}. The basic idea is to take an $\epsilon_0$-approximation $V$ of some gate $U$ and improve it to an $O(\epsilon_0^2)$-approximation of $U$, while taking the length of the approximation from $\ell_0$ to $c \ell_0$ for some constant $c$. Iterating this improvement step allows one to obtain a polylogarithmic overhead for compilation.\footnote{One can easily see the lengths of the gate sequences increase exponentially with each application of this operation, while the error decreases doubly exponentially, which implies the desired polylog dependence of the error.} The key in any proof of a Solovay-Kitaev theorem is to make use of $V$ in this construction in such a way that one does not incur $O(\epsilon_0)$ error in the resulting approximation, as one would naively have from the triangle inequality. In other words, one needs the error in the approximation of $U$ to cancel out to lowest order in $\epsilon_0$.

In the proof of the regular (inverse-closed) Solovay-Kitaev theorem, this is achieved using group commutators \cite{DawsonSolovayKitaev}, which manifestly require inverses in the gate set. Sardharwalla \emph{et al.}~\cite{refocusing} instead achieve this by applying a group averaging function over the Weyl group. They show by direct computation that the lowest order error term in $\epsilon$ cancels out (at least in a neighborhood of the identity).

In our proof, we also consider a group averaging function $f: \SU{d} \to \SU{d}$ based on some (projective) irrep $\rho: G \to \SU{d}$ of a finite group $G$:
\begin{equation}
  f(W) := \prod_{g \in G} \rho(g) W \rho(g)^\dagger.
\end{equation}
Our main technical contribution consists in showing that the lowest order error term cancels here, due to certain orthogonality relations obeyed by irreducible representations.
We show this follows from the fact that the multiplicity of the trivial irrep in the adjoint action of any irrep is one.
Therefore our proof both shows that efficient compilation can occur with a wider family of gate sets than was previously known, and also explains the mathematical reason that Sardharwalla \emph{et al.}'s proof works as it does.

\section{Proof of the main result}

To aid the understanding of our main result, let us first briefly define the relevant notions from representation theory (see \cite{Serre,Childs} for further details).

A $d$-dimensional \emph{representation} of a group $G$ is a map $\rho: G \to \U{d}$ such that $\rho(g_1) \rho(g_2) = \rho(g_1 g_2)$ for all $g_2,g_2 \in G$.
Similarly, $\rho$ is a \emph{projective representation} if it obeys this identity up to a global phase, i.e.\ $\rho(g_1) \rho(g_2) = e^{i \theta(g_1,g_2)} \rho(g_1 g_2)$ for some function $\theta: G \times G \to \R$.
A representation $\rho$ is \emph{reducible} if there is a unitary map $U \in \U{d}$ and two other representations $\rho_1$ and $\rho_2$ of $G$ such that $U \rho(g) U\ct = \rho_1(g) \oplus \rho_2(g)$ for all $g \in G$. If this is not the case, $\rho$ is called \emph{irreducible} (or \emph{irrep} for short).
Finally, if $A \subset B$ are two sets, we say that $A$ is \emph{dense} in $B$ if for any $\varepsilon > 0$ and any $b \in B$ there exists $a \in A$ such that $\norm{a-b} \leq \varepsilon$ for some suitable notion of distance $\norm{\cdot}$.

\ThmMain*

\begin{proof}
By assumption, our gate set is of the form
\begin{equation}
  \G := \rho(G) \cup \set{U_1, U_2, \dotsc, U_N}
\end{equation}
where $\rho(G) := \set{\rho(g) : g \in G}$ and $N \geq 0$ is some integer. Here
\begin{itemize}
  \item $\rho: G \rightarrow \SU{d}$ is a projective irreducible representation of some finite group $G$,
  \item $U_i \in \SU{d}$ are some additional elements whose inverses $U_i\ct$ are not necessarily in $\G$.
\end{itemize}
For the sake of simplicity, we will assume that $\rho$ is an actual irrep rather than a projective irrep (we describe how to generalize the proof to projective irreps in \App{proj}). Note that by the $\G \subset \SU{d}$ assumption we implicitly require that the representation $R$ is in $\SU{d}$ rather than in $\U{d}$. While many irreps are ruled out by this restriction, one can deal with such irreps by first converting them to projective irreps and then applying the techniques discussed in \App{proj}.
We divide the rest of the proof into several steps marked as below.

\partitle{Original gate sequence}
Given a gate $U \in \SU{d}$ which we wish to approximate to accuracy $\epsilon$, we first run the usual Solovay-Kitaev algorithm (see \Sec{Intro}) to obtain a sequence $S_{\epsilon/2}$ of gates whose product $\epsilon/2$-approximates $U$, using elements from $\G$ \emph{and} their inverses. This sequence contains both elements from the set $\rho(G)$ (which is closed under inversion), as well as gates $U_i$ and $U_i\ct$. All of these are in the gate set $\G$ except the $U_i\ct$---and there are only $O\of[\big]{\log^{3.97}(1/\epsilon)}$ many of these.
To prove \Thm{ours}, it therefore suffices to give a Solovay-Kitaev algorithm for approximating the $U_i\ct$ in terms of a sequence of $O(\polylog(1/\epsilon))$ gates from the set $\G$.

\label{footnote:d397}
More concretely, assume we show how to $\epsilon$-approximate each $U_i\ct$ using $O(\log^c(1/\epsilon))$ gates from $\G$ for some constant $c > 0$. Then we can set $\epsilon' := \frac{\epsilon}{2} / O\of[\big]{\log^{3.97}(1/\epsilon)}$ and run this algorithm to $\epsilon'$-approximate each $U_i\ct$ appearing in the sequence $S_{\epsilon/2}$ produced by the regular Solovay-Kitaev algorithm. If we substitute these approximations of $U_i\ct$ back into $S_{\epsilon/2}$, by the triangle inequality the existing error of $\epsilon/2$ in $S_{\epsilon/2}$ will be increased by another $\epsilon/2$ contributed jointly by all $U_i^\dagger$'s. These two contributions together give us the desired $\epsilon$-approximation of $U$.
Note that an $\epsilon'$-approximation of $U_i\ct$ requires $O\of[\big]{\log^{c}(1/\epsilon)}$ gates.\footnote{One can easily see that $\log^c \of[\big]{\frac{\log^{3.97}(1/\epsilon)}{\epsilon}} = O\of[\big]{\log^c(1/\epsilon)}$ as the additional $\log^{3.97}(1/\epsilon)$ factor only adds lower order $\log\log(1/\epsilon)$ terms.}
Hence the $\epsilon$-approximation to $U$ in total will use $O\of[\big]{\log^{c+3.97}(1/\epsilon)}$ gates from $\G$.

\partitle{Initial approximation of $U_i\ct$}
Since $\G$ generates a dense subgroup of $\SU{d}$, there exists a finite length $\ell_0$ such that length-$\ell_0$ sequences of elements of $\G$ are $\epsilon_0$-dense in $\SU{d}$, for a small fixed constant
\begin{equation}
  \epsilon_0 := \frac{1}{6|G|(d-1)! + 2|G|^2}.
  \label{eq:eps0}
\end{equation}
Let us pick among these sequences an initial $\epsilon_0$-approximation of $U_i\ct$ and denote it by $V$. Then
\begin{equation}
  \epsilon_0 \geq \norm{V - U_i\ct} = \norm{VU_i - I},
  \label{eq:VUi}
\end{equation}
where $\norm{\cdot}$ denotes the \emph{operator norm} which is unitarily invariant.

\partitle{Symmetrization}
Now consider the operator $f$ on $\SU{d}$ defined by
\begin{equation}
  f(W) := \displaystyle\prod_{g\in G} \rho(g) W \rho(g)\ct,
  \label{eq:f}
\end{equation}
where the order of the products is taken arbitrarily, as long as the last (rightmost) element of the product corresponds to the identity element $e \in G$.
We are interested in the action of $f$ on $VU_i$.
If we denote the difference in \Eq{VUi} by
$\Op := VU_i - I$
and distribute the product in \Eq{f} into several sums (with no $\Op$'s, with a single copy of $\Op$, two copies of $\Op$, etc.), we get
\begin{align}
  f(VU_i)
  &= \prod_{g \in G} \rho(g) (I+\Op) \rho(g)\ct \\
  &= I
   + \sum_{g \in G} \rho(g) \Op \rho(g)\ct
   + \sum_{\substack{g,g' \in G \\ g < g'}} \rho(g) \Op \rho(g)\ct \rho(g') \Op \rho(g')\ct \\
   &\qquad + \dotsb
   + \prod_{g \in G} \rho(g) \Op \rho(g)\ct,
   \label{eq:fexpanded}
\end{align}
where the order of terms in all products is inherited from \Eq{f} and $g < g'$ refers to this order. Note that the number of terms with $k$ copies of $\Op$ is $\binom{|G|}{k}$.

If one were to naively apply the triangle inequality to this sum, one would obtain that
\begin{equation}
  \norm{f(VU_i)-I} \leq |G|\norm{\Op} + \binom{|G|}{2}\norm{\Op}^2+\ldots
\end{equation}
In other words, one would get that we have moved $f(VU_i)$ further from the identity than we started.
To fix this, we will show that the first term of the above is actually much smaller---of order $\norm{\Op}^2$---and therefore our application of $f$ has moved us closer to the identity.
To see this, first note that using representation theory, one can show that the norm of the first-order term in \Eq{fexpanded} is
\begin{align}
  \norm*{\sum_{g \in G} \rho(g) \Op \rho(g)\ct}
  &= \norm*{\abs{G} \frac{\Tr \Op}{d} I} \\
  &= \abs{G} \frac{\abs{\Tr(V U_i - I)}}{d} . \label{eq:fexpandintermediate}
\end{align}
In other words, the traceless component of the first order term vanishes. This follows from certain orthogonality relations obeyed by irreps, and is proven in \Clm{avg} in \App{claims}.

Next, we show that the trace of $\Op = V U_i - I$ is small compared to its norm, namely
\begin{equation}
  \abs{\Tr(V U_i - I)} \leq (2^d + d!) \norm{V U_i - I}^2.
\end{equation}
This is proven in \Clm{smalltrace} in \App{claims}, and follows essentially because the Lie algebra of the special unitary group is traceless.
Plugging this in to \Eq{fexpandintermediate}, we see that
\begin{align}
  \norm*{\sum_{g \in G} \rho(g) \Op \rho(g)\ct}
  &\leq \abs{G} \frac{2^d + d!}{d} \norm{V U_i - I}^2 \\
  &\leq \abs{G} \frac{2^d + d!}{d} \epsilon_0^2
\end{align}
where we used \Clm{smalltrace} to get the first inequality and \Eq{VUi} to get the second.

Hence, by applying these results and then applying the triangle inequality to \Eq{fexpanded} we get
\begin{align}
  \norm{f(V U_i) - I}
  &\leq \abs{G} \frac{2^d + d!}{d} \epsilon_0^2
  + \sum_{k=2}^{\abs{G}} \epsilon_0^k \binom{\abs{G}}{k} \\
 &\leq \of*{
          \abs{G} \frac{2^d + d!}{d}
          + \frac{\abs{G}^2}{2}
          + \abs{G}^2 \sum_{k=1}^{\abs{G}-2} \epsilon_0^{k} \abs{G}^k} \epsilon_0^2 \label{eq:2ndterm}\\
  &\leq \of*{
          \abs{G} \frac{2^d + d!}{d}
          + \frac{\abs{G}^2}{2}
          + \frac{\abs{G}^2}{2}
        } \epsilon_0^2 \label{eq:geomsum}
\end{align}
Where in \Eq{2ndterm} we used the fact that $\binom{|G|}{2} \leq\frac{|G|^2}{2}$ and $\binom{|G|}{k}\leq |G|^k$, and in \Eq{geomsum} we used the fact that $\epsilon_0 < \frac{1}{2|G|^2}$, so since $|G|>2$ (as G has an irrep of dimension at least 2), we have that $\epsilon_0|G|\leq1/4$ so the geometric sum converges to a quantity $\leq\frac{1}{2}$.

Replacing this with a crude upper bound that $2^d \leq 2d!$ for $d>1$, we get that
\begin{equation}
\norm{f(V U_i) - I} \leq \of*{
          3\abs{G}(d-1)! + \abs{G}^2
        } \epsilon_0^2  =: \epsilon_1
\label{eq:eps1}
\end{equation}
Since we chose $\epsilon_0$ to be $\frac{1}{2(3\abs{G}(d-1)! + \abs{G}^2)}$ in \Eq{eps0}, $\epsilon_1 \leq \frac{\epsilon_0}{2}$ -- in other words $f(VU_i)$ is closer to the identity than $VU_i$.

Multiplying $f(VU_i) - I$ in \Eq{eps1} by $U_i\ct$ on the right, we have that $f(VU_i)U_i\ct$ is an $\epsilon_1$- approximation to $U_i\ct$.
We chose the identity to come last in the definition of $f$ in \Eq{f}, so the string of operators $f(VU_i)$ has the form
\begin{equation}
  f(VU_i)
  = \rho(g_1) VU_i \rho(g_1)\ct
    \rho(g_2) VU_i \rho(g_2)\ct \dotsm VU_i.
\end{equation}
Since $U_iU_i\ct$ cancels at the end, $f(VU_i)U_i\ct$ is an $\epsilon_1$-approximation to $U_i\ct$ using only terms from $\G$.

\partitle{Iterative refinement}
To complete the proof, we iterate this construction by considering
\begin{equation}
  f^{(k)}(VU_i):=f(f(\dotsb f(VU_i))).
\end{equation}
Note from \Eq{eps1} that $f^{(k)}(VU_i) U_i\ct$ is an $\epsilon_k$-approximation to $U_i\ct$, where $\epsilon_k \leq (3\abs{G}(d-1)! + \abs{G}^2)\epsilon_{k-1}^{2}$.
The length of the sequence $f^{(k)}$, denoted $\ell_k$, obeys $\ell_k = |G|\ell_{k-1}+2|G|$.
Again $f^{(k)}(VU_i) U_i\ct$ can be expressed only in terms of elements of $\G$ (since the last $U_i$ in the expansion of $f^{(k)}(VU_i)$ cancels with the rightmost $U_i\ct$ as before).
One can easily show that these recurrence relations imply that as $k$ grows:
\begin{itemize}
  \item the approximation error $\epsilon_k$ shrinks doubly exponentially: $\epsilon_k \leq \frac{2\epsilon_0}{2^{2^k}}$;
  \item the length of the sequence $\ell_k$ grows exponentially: $\ell_k = O(|G|^k \ell_0)$.
\end{itemize}
Note that this sort of asymptotic behavior occurs simply because $\epsilon_k = O(\epsilon_{k-1}^2)$ while $\ell_k = O(\ell_{k-1})$ (though of course the value of $\epsilon_0$ used in the recurrence may depend on the hidden constant in the big-O notation).
This immediately implies that one can approximate $U_i\ct$ to accuracy $\epsilon$ with merely polylog overhead, as desired.
More specifically, such approximation uses
\begin{equation}
  O\of[\big]{\ell_0 \log^{\log_2|G|}(1/\epsilon)}
  \label{eq:scaling}
\end{equation}
elements of $\G$. By our analysis at the beginning of the proof, this gives a Solovay-Kitaev theorem with an exponent of $\log_2|G|+3.97$ in the polylog, completing the proof of \Thm{ours}.
\end{proof}

We have therefore shown that one can $\epsilon$-approximate any $U_i^\dagger$ using only gates from our gate set $\G$ using merely $\polylog(1/\epsilon)$ gates.
The exponent of the polylog for approximating each $U_i^\dagger$ is again easily computed to be $O(\log_2 |G|)$. So putting this all together, our approximation for the overall unitary $U$ requires
\begin{equation}
  O\of[\big]{\log^{\log_2 |G|}(1/\epsilon)}
\end{equation}
gates from $\G$.
Note that the dependence on dimension $d$ and order of the group $G$ is hidden in the big-O notation, which hides a factor of $\ell_0$, the length of sequences required to achieve an initial $\epsilon_0$-net of $\SU{d}$. By a volume argument $\ell_0= \Omega(d^2)$ \cite{DawsonSolovayKitaev}.
In fact our choice of $\epsilon_0$  implies that $\ell_0=\Omega( d^3\log d)$ in our construction.\footnote{Since an $\epsilon_0$-ball occupies $\Theta(\epsilon_0^{d^2})$ volume in $\SU{d}$, $\ell_0 = \Omega\of[\big]{d^2 \log (1/\epsilon_0)}$ \cite{DawsonSolovayKitaev}.
Since we set $\epsilon_0 = (2|G|^2+6|G|(d-1)!)^{-1}$ in \Eq{eps0}, we have that $\ell_0 = \Omega\of[\big]{d^3\log d}$ since $\log d!$ scales as $O(d\log d)$ by Stirling's formula.}

\subsection{Extensions of our theorem}

We have shown a Solovay-Kitaev theorem for any gate set $\G$ that contains an irrep of a finite group $G$, without requiring $\G$ to be inverse-closed. Our result can be easily generalized in two directions.

First, our proof also works if instead of an irrep we have a \emph{projective} irrep. That is, a map $\rho: G \to \SU{d}$ such that, for any $g_1, g_2 \in G$,
\begin{equation}
  \rho(g_1) \rho(g_2) = e^{i\theta(g_1,g_2)} \rho(g_1 g_2)
\end{equation}
for some collection of phases\footnote{The quantity $e^{i\theta(g_1,g_2)}$ is also known as a \emph{Schur multiplier} of $G$.} $\theta(g_1,g_2) \in [0,2\pi)$. In such case one still has a Solovay-Kitaev theorem for any universal gate set that includes $\rho(G)$.
For instance, the Pauli matrices $\{I,X,Y,Z\}$ form a projective irrep, but not an irrep (though the matrices $\{\pm1,\pm i\}\cdot\{I,X,Y,Z\}$ do form an irrep).
Since the exponent of the logarithm of our version of the Solovay-Kitaev theorem contains $\log_2 |G|$, this generalization improves the exponent
(e.g.\ using the four Pauli matrices instead of the eight-element Pauli group improves the exponent by an additive 2).
We give details on why projective irreps suffice in \App{proj}.

Second, we note that our proof can be extended to the \emph{special linear group} $\SL{d}$ as well. That is, one can also efficiently compile non-singular matrices, so long as a (projective) irrep is present in a gate set that is universal for $\SL{d}$.
A Solovay-Kitaev Theorem (with inverses) for the special linear group was first shown by Aharnov, Arad, Eban and Landau \cite{aharonov2007polynomial}, who used it to prove that additive approximations to the Tutte polynomial are $\textsf{BQP}$-hard in many regimes.
It was also applied by \cite{boulandCCC2016} to the problem of classifying quantum gate sets, where it arose naturally because the ``postselection gadgets'' used in their proof are non-unitary.
For a formal description of the non-unitary version of this theorem, please see \App{SLnC}.
Since postselection gadgets are often non-unitary \cite{boulandCCC2016}, we likewise expect this version of the theorem will be more useful for gate classification problems.

\section{Open problems} \label{sec:Problems}

The main unresolved problem left by our work is to prove a generic inverse-free Solovay-Kitaev theorem, which has been a longstanding open problem \cite{DawsonSolovayKitaev,kuperberg2009hard}.

\begin{conjecture*}[Inverse-free Solovay-Kitaev theorem] \label{invfreesk} For any fixed $d\geq2$, if $\G \subset \SU{d}$ is a finite gate set which densely generates $\SU{d}$, then there is an algorithm which outputs an $\epsilon$-approximation to any $U \in \SU{d}$ using merely $O(\polylog(1/\epsilon))$ elements from $\G$.
\end{conjecture*}
One can easily see that for any universal gate set (possibly without inverses), one can $\epsilon$-approximate arbitrary unitaries with $O(1/\epsilon)$ overhead. This follows from simply running the Solovay-Kitaev theorem with inverses, and then approximating each inverse $W^\dagger$ with $W^k$ for some integer $k$ (which one can do with $O(1/\epsilon)$ overhead as this is simply composing irrational rotations about a single axis).
However current approaches seem to be unable to improve this compilation algorithm from $O(1/\epsilon)$ to $\polylog(1/\epsilon)$.
As discussed in \Sec{prooftechniques}, current proofs of the Solovay-Kitaev theorem require a special cancellation of error terms in order to convert an $\epsilon$-approximation of some operator into an $O(\epsilon^2)$-approximation. This cancellation of error terms can be achieved by taking group commutators \cite{DawsonSolovayKitaev} or, as in this work and \cite{refocusing}, it can be achieved by averaging over irreps and using the orthogonality of irreps.
However, there is no known technique for achieving this sort of error cancellation without having some structure in the gate set.\footnote{For example, Zhiyenbayev, Akulin, and Mandilara \cite{katerina2017diffuse} have recently studied an alternative setting where instead of inverses a certain ``isotropic'' property of the gates is assumed.}

Additionally, a natural question is whether the value of $\epsilon_0$ can be improved. This would improve the scaling of our result with dimension.
In our result (and in the inverse-closed Solovay-Kitaev Theorem) the big-$O$ notation hides a factor of $\ell_0$---the length of the initial sequences required to achieve an $\epsilon_0$-net.
In our result $\epsilon_0$ scales as $1/d!$, and hence a volume argument implies $\ell_0=\Omega(d^3\log d)$.
In contrast the (inverse-closed) Solovay-Kitaev theorem merely requires $\epsilon_0=\Theta(1)$ resulting in $\ell_0=\Omega(d^2$)~\cite{DawsonSolovayKitaev}.
It is a natural question if one can improve the value of $\epsilon_0$ and therefore improve dimension dependence of our construction.

A somewhat simpler open problem is whether our theorem can be improved by considering particular orders of the group elements in \Eq{f}. The function $f(U)$ which we iterate when proving \Thm{ours} is defined by averaging over the irrep of $G$ in an arbitrary order; our theorem essentially works because if $U$ is $\epsilon$-close to the identity then $f(U)$ is $O(\epsilon^2)$-close to the identity. However, we have found by direct calculation that for the 2-dimensional irrep of $S_3$, considering particular orders of the group can lead to the $O(\epsilon^2)$ terms cancelling out as well, leaving only  $O(\epsilon^3)$ terms. It is an interesting open problem if these additional cancellations can be generalized to other groups. If so, they would improve the $\log_2 |G|$ in the exponent of the logarithm of our result to $\log_k |G|$, where $k$ is the lowest order remaining error term.

Finally, we note one may be able to extend our results to compilation over more general Lie groups, just as Kuperberg extended the inverse-closed Solovay-Kitaev theorem to arbitrary connected Lie groups whose Lie algebra is perfect \cite{kuperberg2009hard}. We leave this as an open problem.

\section*{Acknowledgments}
We thank Aram Harrow, Vadym Kliuchnikov, and Zolt\'{a}n Zimbor\'{a}s for helpful discussions, and Adam Sawicki for pointing us to reference \cite{varju2013random}.
AB was partially supported by the NSF GRFP under Grant No. 1122374, by the Vannevar Bush Fellowship from the US Department of Defense, by ARO Grant W911NF-12-1-0541, by NSF Grant CCF-1410022, and by the NSF Waterman award under grant number 1249349.
Part of this work was done while MO was at the University of Cambridge where he was supported by a Leverhulme Trust Early Career Fellowship (ECF-2015-256). MO also acknowledges hospitality of MIT where this project was initiated.


\begin{thebibliography}{CLMO11}

\bibitem[AA11]{aaronsonboson}
Scott Aaronson and Alex Arkhipov.
\newblock The computational complexity of linear optics.
\newblock In {\em Proceedings of the forty-third annual ACM symposium on Theory
  of computing}, pages 333--342. ACM, 2011.
\newblock \href {http://arxiv.org/abs/1011.3245} {\path{arXiv:1011.3245}},
  \href {http://dx.doi.org/10.1145/1993636.1993682}
  {\path{doi:10.1145/1993636.1993682}}.

\bibitem[AAEL07]{aharonov2007polynomial}
Dorit Aharonov, Itai Arad, Elad Eban, and Zeph Landau.
\newblock Polynomial quantum algorithms for additive approximations of the
  {P}otts model and other points of the {T}utte plane.
\newblock 2007.
\newblock \href {http://arxiv.org/abs/quant-ph/0702008}
  {\path{arXiv:quant-ph/0702008}}.

\bibitem[ABO08]{aharonov1997fault}
Dorit Aharonov and Michael Ben-Or.
\newblock Fault-tolerant quantum computation with constant error rate.
\newblock {\em SIAM Journal on Computing}, 38(4):1207--1282, 2008.
\newblock \href {http://arxiv.org/abs/quant-ph/9906129}
  {\path{arXiv:quant-ph/9906129}}, \href
  {http://dx.doi.org/10.1137/S0097539799359385}
  {\path{doi:10.1137/S0097539799359385}}.

\bibitem[AJL09]{aharonov2009polynomial}
Dorit Aharonov, Vaughan Jones, and Zeph Landau.
\newblock A polynomial quantum algorithm for approximating the {J}ones
  polynomial.
\newblock {\em Algorithmica}, 55(3):395--421, Nov 2009.
\newblock \href {http://arxiv.org/abs/quant-ph/0511096}
  {\path{arXiv:quant-ph/0511096}}, \href
  {http://dx.doi.org/10.1007/s00453-008-9168-0}
  {\path{doi:10.1007/s00453-008-9168-0}}.

\bibitem[BCK15]{berry2015hamiltonian}
Dominic~W. Berry, Andrew~M. Childs, and Robin Kothari.
\newblock Hamiltonian simulation with nearly optimal dependence on all
  parameters.
\newblock In {\em Foundations of Computer Science (FOCS), 2015 IEEE 56th Annual
  Symposium on}, pages 792--809. IEEE, Oct 2015.
\newblock \href {http://arxiv.org/abs/1501.01715} {\path{arXiv:1501.01715}},
  \href {http://dx.doi.org/10.1109/FOCS.2015.54}
  {\path{doi:10.1109/FOCS.2015.54}}.

\bibitem[BCO17]{BCO16}
Johannes Bausch, Toby Cubitt, and Maris Ozols.
\newblock The complexity of translationally invariant spin chains with low
  local dimension.
\newblock {\em Annales Henri Poincar{\'e}}, 18(11):3449--3513, Nov 2017.
\newblock \href {http://arxiv.org/abs/1605.01718} {\path{arXiv:1605.01718}},
  \href {http://dx.doi.org/10.1007/s00023-017-0609-7}
  {\path{doi:10.1007/s00023-017-0609-7}}.

\bibitem[BFK18]{cccs}
Adam Bouland, Joseph~F. Fitzsimons, and Dax~E. Koh.
\newblock Complexity classification of conjugated {C}lifford circuits.
\newblock {\em Proc. 33rd Computational Complexity Conference (CCC)}, 2018.
\newblock \href {http://arxiv.org/abs/1709.01805} {\path{arXiv:1709.01805}}.

\bibitem[BJS10]{BJS2010}
Michael~J. Bremner, Richard Jozsa, and Dan~J. Shepherd.
\newblock Classical simulation of commuting quantum computations implies
  collapse of the polynomial hierarchy.
\newblock {\em Proceedings of the Royal Society of London A: Mathematical,
  Physical and Engineering Sciences}, 2010.
\newblock \href {http://arxiv.org/abs/1005.1407} {\path{arXiv:1005.1407}},
  \href {http://dx.doi.org/10.1098/rspa.2010.0301}
  {\path{doi:10.1098/rspa.2010.0301}}.

\bibitem[BMZ16]{boulandCCC2016}
Adam Bouland, Laura Man\v{c}inska, and Xue Zhang.
\newblock Complexity classification of two-qubit commuting {H}amiltonians.
\newblock In Ran Raz, editor, {\em 31st Conference on Computational Complexity
  (CCC 2016)}, volume~50 of {\em Leibniz International Proceedings in
  Informatics (LIPIcs)}, pages 28:1--28:33, Dagstuhl, Germany, 2016. Schloss
  Dagstuhl--Leibniz-Zentrum f{\"u}r Informatik.
\newblock \href {http://arxiv.org/abs/1602.04145} {\path{arXiv:1602.04145}},
  \href {http://dx.doi.org/10.4230/LIPIcs.CCC.2016.28}
  {\path{doi:10.4230/LIPIcs.CCC.2016.28}}.

\bibitem[BP17]{BP17}
Johannes Bausch and Stephen Piddock.
\newblock The complexity of translationally invariant low-dimensional spin
  lattices {3D}.
\newblock {\em Journal of Mathematical Physics}, 58(11):111901, 2017.
\newblock \href {http://arxiv.org/abs/1702.08830} {\path{arXiv:1702.08830}},
  \href {http://dx.doi.org/10.1063/1.5011338} {\path{doi:10.1063/1.5011338}}.

\bibitem[BRS15a]{msr2015fallback}
Alex Bocharov, Martin Roetteler, and Krysta~M. Svore.
\newblock Efficient synthesis of probabilistic quantum circuits with fallback.
\newblock {\em Phys. Rev. A}, 91(5):052317, May 2015.
\newblock \href {http://arxiv.org/abs/409.3552} {\path{arXiv:409.3552}}, \href
  {http://dx.doi.org/10.1103/PhysRevA.91.052317}
  {\path{doi:10.1103/PhysRevA.91.052317}}.

\bibitem[BRS15b]{bocharov2015efficient}
Alex Bocharov, Martin Roetteler, and Krysta~M. Svore.
\newblock Efficient synthesis of universal repeat-until-success quantum
  circuits.
\newblock {\em Phys. Rev. Lett.}, 114(8):080502, Feb 2015.
\newblock \href {http://arxiv.org/abs/1404.5320} {\path{arXiv:1404.5320}},
  \href {http://dx.doi.org/10.1103/PhysRevLett.114.080502}
  {\path{doi:10.1103/PhysRevLett.114.080502}}.

\bibitem[Chi13]{Childs}
Andrew~M. Childs.
\newblock Fourier analysis in nonabelian groups.
\newblock Lecture notes at University of Waterloo, 2013.
\newblock URL: \url{http://www.cs.umd.edu/~amchilds/teaching/w13/l06.pdf}.

\bibitem[Chi17]{ChildsQA}
Andrew~M. Childs.
\newblock Lecture notes on quantum algorithms.
\newblock Lecture notes at University of Maryland, 2017.
\newblock URL: \url{http://www.cs.umd.edu/~amchilds/qa/qa.pdf}.

\bibitem[CLMO11]{CLMO11}
Andrew~M. Childs, Debbie Leung, Laura Man{\v{c}}inska, and Maris Ozols.
\newblock Characterization of universal two-qubit {H}amiltonians.
\newblock {\em Quantum Information \& Computation}, 11(1\&2):19--39, Jan 2011.
\newblock URL: \url{http://www.rintonpress.com/journals/qiconline.html#v11n12},
  \href {http://arxiv.org/abs/1004.1645} {\path{arXiv:1004.1645}}.

\bibitem[DN06]{DawsonSolovayKitaev}
Christopher~M. Dawson and Michael~A. Nielsen.
\newblock The {S}olovay-{K}itaev algorithm.
\newblock {\em Quantum Information \& Computation}, 6(1):81--95, Jan 2006.
\newblock URL: \url{http://www.rintonpress.com/journals/qiconline.html#v6n1},
  \href {http://arxiv.org/abs/quant-ph/0505030}
  {\path{arXiv:quant-ph/0505030}}.

\bibitem[GI13]{GI13}
Daniel Gottesman and Sandy Irani.
\newblock The quantum and classical complexity of translationally invariant
  tiling and {H}amiltonian problems.
\newblock {\em Theory of Computing}, 9(2):31--116, 2013.
\newblock \href {http://arxiv.org/abs/0905.2419} {\path{arXiv:0905.2419}},
  \href {http://dx.doi.org/10.4086/toc.2013.v009a002}
  {\path{doi:10.4086/toc.2013.v009a002}}.

\bibitem[HRC02]{harrowrechtchuang}
Aram~W. Harrow, Benjamin Recht, and Isaac~L. Chuang.
\newblock Efficient discrete approximations of quantum gates.
\newblock {\em Journal of Mathematical Physics}, 43(9):4445--4451, 2002.
\newblock \href {http://arxiv.org/abs/quant-ph/0111031}
  {\path{arXiv:quant-ph/0111031}}, \href {http://dx.doi.org/10.1063/1.1495899}
  {\path{doi:10.1063/1.1495899}}.

\bibitem[Kit97]{Kitaev}
Alexei~Yu. Kitaev.
\newblock Quantum computations: algorithms and error correction.
\newblock {\em Russian Mathematical Surveys}, 52(6):1191--1249, 1997.
\newblock \href {http://dx.doi.org/10.1070/RM1997v052n06ABEH002155}
  {\path{doi:10.1070/RM1997v052n06ABEH002155}}.

\bibitem[Kli13]{kliuchnikov2013synthesis}
Vadym Kliuchnikov.
\newblock Synthesis of unitaries with {C}lifford+{T} circuits.
\newblock 2013.
\newblock \href {http://arxiv.org/abs/1306.3200} {\path{arXiv:1306.3200}}.

\bibitem[KMM16]{kliuchnikov2016practical}
Vadym Kliuchnikov, Dmitri Maslov, and Michele Mosca.
\newblock Practical approximation of single-qubit unitaries by single-qubit
  quantum {C}lifford and {T} circuits.
\newblock {\em IEEE Transactions on Computers}, 65(1):161--172, Jan 2016.
\newblock \href {http://arxiv.org/abs/1212.6964} {\path{arXiv:1212.6964}},
  \href {http://dx.doi.org/10.1109/TC.2015.2409842}
  {\path{doi:10.1109/TC.2015.2409842}}.

\bibitem[KSV02]{Kitaev-book}
Alexei~Yu. Kitaev, Alexander Shen, and Mikhail~N. Vyalyi.
\newblock {\em Classical and Quantum Computation}, volume~47 of {\em Graduate
  Studies in Mathematics}.
\newblock American Mathematical Society, 2002.
\newblock URL: \url{https://books.google.com/books?id=qYHTvHPvmG8C}.

\bibitem[Kup15]{kuperberg2009hard}
Greg Kuperberg.
\newblock How hard is it to approximate the {J}ones polynomial?
\newblock {\em Theory of Computing}, 11(6):183--219, 2015.
\newblock \href {http://arxiv.org/abs/0908.0512} {\path{arXiv:0908.0512}},
  \href {http://dx.doi.org/10.4086/toc.2015.v011a006}
  {\path{doi:10.4086/toc.2015.v011a006}}.

\bibitem[NC10]{nielsen2002quantum}
Michael~A. Nielsen and Isaac~L. Chuang.
\newblock {\em Quantum Computation and Quantum Information}.
\newblock Cambridge University Press, 2010.
\newblock URL: \url{https://books.google.com/books?id=-s4DEy7o-a0C}.

\bibitem[Ozo09]{MarisSK}
M\={a}ris Ozols.
\newblock The {S}olovay-{K}itaev theorem.
\newblock Essay at University of Waterloo, 2009.
\newblock URL:
  \url{http://home.lu.lv/~sd20008/papers/essays/Solovay-Kitaev.pdf}.

\bibitem[PS14]{paetznick2014repeat}
Adam Paetznick and Krysta~M. Svore.
\newblock Repeat-until-success: {N}on-deterministic decomposition of
  single-qubit unitaries.
\newblock {\em Quantum Information \& Computation}, 14(15\&16):1277--1301,
  2014.
\newblock URL:
  \url{http://www.rintonpress.com/journals/qiconline.html#v14n1516}, \href
  {http://arxiv.org/abs/1311.1074} {\path{arXiv:1311.1074}}.

\bibitem[PS18]{parzanchevski2017super}
Ori Parzanchevski and Peter Sarnak.
\newblock {S}uper-{G}olden-{G}ates for {PU}(2).
\newblock {\em Advances in Mathematics}, 327:869--901, 2018.
\newblock Special volume honoring David Kazhdan.
\newblock \href {http://arxiv.org/abs/1704.02106} {\path{arXiv:1704.02106}},
  \href {http://dx.doi.org/10.1016/j.aim.2017.06.022}
  {\path{doi:10.1016/j.aim.2017.06.022}}.

\bibitem[RS16]{ross2014optimal}
Neil~J. Ross and Peter Selinger.
\newblock Optimal ancilla-free {C}lifford+{T} approximation of {Z}-rotations.
\newblock {\em Quantum Information \& Computation}, 16(11\&12):0901--0953,
  2016.
\newblock URL:
  \url{http://www.rintonpress.com/journals/qiconline.html#v16n1112}, \href
  {http://arxiv.org/abs/1403.2975} {\path{arXiv:1403.2975}}, \href
  {http://dx.doi.org/10.26421/QIC16.11-12} {\path{doi:10.26421/QIC16.11-12}}.

\bibitem[Sar15]{sarnak2015letter}
Peter Sarnak.
\newblock Letter to {A}aronson and {P}ollington on the {S}olvay-{K}itaev
  {T}heorem and {G}olden {G}ates, 2015.
\newblock URL: \url{http://publications.ias.edu/sarnak/paper/2637}.

\bibitem[SCHL16]{refocusing}
Imdad~S.B. Sardharwalla, Toby~S. Cubitt, Aram~W. Harrow, and Noah Linden.
\newblock Universal refocusing of systematic quantum noise.
\newblock 2016.
\newblock \href {http://arxiv.org/abs/1602.07963} {\path{arXiv:1602.07963}}.

\bibitem[Sel15]{selinger2015efficient}
Peter Selinger.
\newblock Efficient {C}lifford+{T} approximation of single-qubit operators.
\newblock {\em Quantum Information \& Computation}, 15(1\&2):159--180, 2015.
\newblock URL: \url{http://www.rintonpress.com/journals/qiconline.html#v15n12},
  \href {http://arxiv.org/abs/1212.6253} {\path{arXiv:1212.6253}}.

\bibitem[Ser12]{Serre}
Jean-Pierre Serre.
\newblock {\em Linear Representations of Finite Groups}.
\newblock Springer, 2012.
\newblock URL: \url{https://books.google.com/books?id=9mT1BwAAQBAJ}.

\bibitem[Sho97]{shor1999polynomial}
Peter~W. Shor.
\newblock Polynomial-time algorithms for prime factorization and discrete
  logarithms on a quantum computer.
\newblock {\em SIAM Journal on Computing}, 26(5):1484--1509, 1997.
\newblock \href {http://arxiv.org/abs/quant-ph/9508027}
  {\path{arXiv:quant-ph/9508027}}, \href
  {http://dx.doi.org/10.1137/S0097539795293172}
  {\path{doi:10.1137/S0097539795293172}}.

\bibitem[SK17]{SK17}
Adam Sawicki and Katarzyna Karnas.
\newblock Universality of single-qudit gates.
\newblock {\em Annales Henri Poincar{\'e}}, Aug 2017.
\newblock \href {http://arxiv.org/abs/1609.05780} {\path{arXiv:1609.05780}},
  \href {http://dx.doi.org/10.1007/s00023-017-0604-z}
  {\path{doi:10.1007/s00023-017-0604-z}}.

\bibitem[Var13]{varju2013random}
P{\'e}ter~P{\'a}l Varj{\'u}.
\newblock Random walks in compact groups.
\newblock {\em Documenta Mathematica}, 18:1137--1175, 2013.
\newblock URL:
  \url{https://www.math.uni-bielefeld.de/documenta/vol-18/35.html}, \href
  {http://arxiv.org/abs/1209.1745} {\path{arXiv:1209.1745}}.

\bibitem[ZAM17]{katerina2017diffuse}
Y.~Zhiyenbayev, V.~M. Akulin, and A.~Mandilara.
\newblock Quantum compiling with diffusive sets of gates.
\newblock 2017.
\newblock \href {http://arxiv.org/abs/1708.08909} {\path{arXiv:1708.08909}}.

\end{thebibliography}

\appendix

\section{Auxiliary claims} \label{app:claims}

\begin{claim}\label{claim:avg}
If $\rho$ is a $d$-dimensional (projective) irrep of some finite group $G$ and $M$ is any $d \times d$ complex matrix then
\begin{equation}
  \sum_{g \in G} \rho(g) M \rho(g)\ct = \abs{G} \frac{\Tr M}{d} I.
\end{equation}
\end{claim}

\newcommand{\ket}[1]{|#1\rangle}
\newcommand{\bra}[1]{\langle#1|}

\begin{proof}
If $\rho$ and $\rho'$ are any two irreps of a finite group $G$, with dimensions $d_\rho$ and $d_{\rho'}$ respectively, their matrix entries obey the following orthogonality relations~\cite{Serre}:
\begin{equation}
  \frac{d_\rho}{\abs{G}} \sum_{g \in G}
  \rho(g)_{ij} \, \overline{\rho'(g)_{kl}}
  = \delta_{\rho\rho'} \delta_{ik} \delta_{jl}, \qquad
  \forall i,j \in \set{1, \dotsc, d_\rho}, \quad
  \forall k,l \in \set{1, \dotsc, d_{\rho'}}.
  \label{eq:orth1}
\end{equation}
In particular, if $\rho = \rho'$ and we write the matrix entries as $\rho(g)_{ij} = \bra{i} \rho(g) \ket{j}$ then
\begin{equation}
  \frac{d}{\abs{G}} \sum_{g \in G}
  \bra{i} \rho(g) \ket{j} \bra{l} \rho(g)\ct \ket{k}
  = \delta_{ik} \delta_{jl}, \qquad
  \forall i,j,k,l \in \set{1, \dotsc, d}
  \label{eq:orth2}
\end{equation}
where $d := d_\rho = d_{\rho'}$.
If we multiply both sides by $\ket{i}\bra{k}$ and then sum over $i$ and $k$, we get
\begin{equation}
  \frac{d}{\abs{G}} \sum_{g \in G}
  \rho(g) \ket{j} \bra{l} \rho(g)\ct
  = I \delta_{jl}, \qquad
  \forall j,l \in \set{1, \dotsc, d}.
  \label{eq:jl}
\end{equation}
If $M = \sum_{j,l=1}^d m_{jl} \ket{j} \bra{l}$ then by linearity,
\begin{equation}
  \frac{d}{\abs{G}}
  \sum_{g \in G} \rho(g) M \rho(g)\ct
  = I \sum_{j,l=1}^d m_{jl} \delta_{jl}
  = I \Tr M,
\end{equation}
which completes the proof.
\end{proof}

Another way to see this result is by noticing that the adjoint action of $\rho$ decomposes as a direct sum of the trivial representation (acting on the $1$-dimensional space spanned by the identity matrix) and a $(d^2-1)$-dimensional representation without any trivial component.
This follows from Schur's first lemma.
The result then follows by the orthogonality relations obeyed by the irrep decomposition of the adjoint action.

\begin{claim}\label{claim:smalltrace}
If $M \in \SL{d}$ then $\abs{\Tr M - d} \leq (2^d + d!) \norm{M - I}^2$.
\end{claim}

\begin{proof}
Let $A := M - I$ and denote the entries of $A$ by $a_{ij}$ where $i , j = 1, \dotsc, d$. We know that $1 = \det M = \det (A + I)$, so expanding in terms of the $a_{ij}$'s, we have that
\begin{equation}
  1 =
  \sum_{\sigma \in \S{d}} \sgn(\sigma)
  \prod_{i=1}^d \of{ a_{i\sigma(i)} + \delta_{i\sigma(i)} }.
\end{equation}
Now let us simply take out the term with $\sigma = \idp$, the identity permutation:
\begin{equation}
  1
  = \prod_{i=1}^d \of{a_{ii} + 1}
  + \sum_{\sigma \in \S{d} \setminus \set{\idp}}
    \sgn(\sigma)
    \prod_{i=1}^d \of{ a_{i\sigma(i)} + \delta_{i\sigma(i)} }.
\end{equation}
And now expanding the first term we see
\begin{equation}
  1 = 1
  + \sum_{i=1}^d a_{ii}
  + \sum_{i \neq j} a_{ii} a_{jj}
  + \dotsb
  + a_{11} a_{22} \dotsm a_{dd}
  + \sum_{\sigma \in \S{d} \setminus \set{\idp}} \sgn(\sigma)
    \prod_{i=1}^d \of{ a_{i\sigma(i)} + \delta_{i\sigma(i)} },
\end{equation}
which implies
\begin{equation}
  - \Tr A
  = \sum_{i \neq j} a_{ii} a_{jj} + \dotsb + a_{11} a_{22} \dotsm a_{dd}
  + \sum_{\sigma \in \S{d} \setminus \set{\idp}} \sgn(\sigma)
    \prod_{i=1}^d \of{ a_{i\sigma(i)} + \delta_{i\sigma(i)} }.
\end{equation}

Now observe that each of the terms on the right hand side is quadratic in the $a_{ij}$'s---this is because any non-identity permutation displaces at least two items. Let $c \leq 2^d + d!$ denote the number of the terms present, which is constant in any fixed dimension $d$. Hence we have that
\begin{equation}
  \abs{\Tr M - d}
  = \abs{\Tr A}
  \leq c \max_{i,j} \; \abs{a_{ij}}^2
  \leq c \norm{A}^2
  = c \norm{M - I}^2
\end{equation}
where we used $\abs{a_{ij}} \leq \norm{A}$ in the last inequality (this follows by choosing the $j$-th standard basis vector in the definition of the operator norm).
\end{proof}

Note that this claim, i.e.\ that elements $\epsilon$-close to the identity have trace substantially smaller than $\epsilon$, is a reflection of the fact that the Lie algebra of the special linear group is traceless.

\section{Representations vs projective representations} \label{app:proj}

Throughout our proof of \Thm{ours}, we assumed that $\rho$ is an irrep of the group $G$. Here we show that the same construction works also for a projective irrep of $G$. In other words, even if $\rho(g_1) \rho(g_2) = e^{i\theta(g_1,g_2)} \rho(g_1 g_2)$ for some phase $\theta(g_1, g_2) \in [0,2\pi)$, our version of the Solovay-Kitaev theorem still holds.
As the Weyl operators merely form a projective representation, this allows our result to strictly generalize that of \cite{refocusing}.
Intuitively, such generalization is to be expected since global phases are non-physical in quantum theory. We make this precise below.

Suppose that we have a projective representation $\rho$ of a finite group $G$. It is convenient to think of $\rho(G)$ as a subset of the \emph{projective unitary group} $\PU{d}$ that consists of equivalence classes of elements of $\U{d}$ that differ only by global phase. Note that $\PU{d} = \PSU{d}$, the \emph{special} projective unitary group, since $\overline{\det(U)} U \in \SU{d}$ for any $U \in \U{d}$. Now, consider extending the projective representation $\rho$ in $\PSU{d}$ into a representation\footnote{This is known as a \emph{central extension} of the representation.} $\rho'$ in $\SU{d}$. Since
\begin{equation}
  \PSU{d} = \SU{d}/\Z_d,
\end{equation}
i.e.\ the only difference between projective and non-projective representations are factors of $e^{2\pi i/d}I$, this merely increases the size of the group by an integer multiple $k$ which is a divisor of $d$. Let us denote this larger group by $G'$.

Now consider applying our proof of \Thm{ours} to $\rho'$ and $G'$. The corresponding averaging operator is
\begin{equation}
  f'(W) := \prod_{g \in G'} \rho'(g) W \rho'(g)\ct.
\end{equation}
Our proof essentially uses two facts:
\begin{enumerate}
\item The trace of $W$ is small relative to its distance from the identity (\Clm{smalltrace}).
\item The traceless component of $W$ vanishes to lowest order because from \Clm{avg} we have that for any tracless $\Op$,
\begin{equation}
  \sum_{g \in G'} \rho'(g) \Op \rho'(g)\ct = 0.
  \label{eq:sumG'}
\end{equation}
\end{enumerate}
Note that if $g, h \in G'$ are such that $\rho'(g) = e^{i\theta} \rho'(h)$ for some $\theta \in \R$, then they contribute identical terms in the above sum, since the global phase factors commute through and cancel out.
Since the any projectively equivalent group elements $g,h$ contribute the same quantity to the sum, and $G'$ is simply a (projective) $k$-fold cover of $G$, this means that we can rewrite \Eq{sumG'} as
\begin{equation}
  k \sum_{g \in G} \rho'(g) \Op \rho'(g)\ct = 0,
\end{equation}
where we have simply summed over one representative from each set of projectively equivalent representatives.

Therefore, if we had instead considered averaging over the projective representation only using the original averaging operator (which involves a factor $k$ fewer products),
\begin{equation}
  f(W) := \prod_{g \in G} \rho(g) W \rho(g)\ct,
\end{equation}
the corresponding sum in \Eq{sumG'} (which is the above sum divided by $k$) would be 0 as well.
Therefore, the cancellation of lowest-order terms for the traceless component of the error---i.e. the second fact listed above---still holds.
Furthermore, the first fact is true independent of the group $G$ considered, and is simply a fact about matrices of determinant 1 which are close to the identity.
Therefore, the proof of \Thm{ours} works exactly as before if $\rho$ is merely a projective representation.

\section{Extension to the special linear group}
\label{app:SLnC}

In this appendix we describe how to extend our proof of Theorem \ref{thm:ours} to the non-unitary case.
Namely, we want to approximate some matrix $M \in \SL{d}$, our gate set $\G \subset \SL{d}$ is dense in $\SL{d}$, and it contains a (possibly non-unitary) irrep of a finite group $G$ as well as some additional gates $U_i\in \SL{d}$.

Let us argue that an $\epsilon$-approximation of $M$ can be obtained using the same algorithm as in the proof of \Thm{ours}, but with one minor change.
Namely, in the first step of the algorithm one must apply the non-unitary Solovay-Kitaev Theorem (with inverses) of Aharonov, Arad, Eban, and Landau \cite{aharonov2007polynomial} rather than the usual unitary Solovay-Kitaev Theorem (with inverses).
As before, the problem therefore reduces to finding an expression for the elements $U_i^{-1}$ in terms of $\G$.
Note that no other step of our proof requires any matrices to be unitary!
Recall that the heart of the proof was in showing that if $VU_i$ is $\epsilon$-close to $I$ then $f(VU_i)$ is $O(\epsilon^2)$-close to $I$, where $V$ denotes the initial $\epsilon_0$-approximation of $U_i^{-1}$. The key facts that we used to show this are:
\begin{itemize}
\item \Clm{avg}, which states that the traceless component of $VU_i$ vanishes to first order under the application of $f$ due to the orthogonality of irreps.
\item \Clm{smalltrace}, which states that matrices of determinant 1 which are $\epsilon$-close to the identity have trace $O(\epsilon^2)$.
\end{itemize}
Neither of these depends on the matrices involved being unitary---indeed the Schur orthogonality relations between irreps in \Eq{orth1} also hold for non-unitary irreps.
Therefore, our proof implies the following:

\begin{theorem}\label{thm:oursSL}
For any fixed $d\geq2$, suppose $\G \subset \SL{d}$ is a finite gate set that contains a (projective) irrep of some finite group $G$.
Let $r > 0$ be any fixed radius, let $B_r$ be the ball of radius $r$ about the identity in $\SL{d}$, and suppose that $\G$ densely generates all transformations in $B_r$.
Then there is an algorithm which outputs an $\epsilon$-approximation to any $M \in B_r$ using merely $O(\polylog(1/\epsilon))$ elements from $\G$.
\end{theorem}

Other than the replacement of $\SU{d}$ with $\SL{d}$, the only thing that differs between this theorem and \Thm{ours} is the additional restriction that the matrix $M$ we are approximating is a finite distance from the identity (as is present in the non-unitary Solovay-Kitaev theorem of \cite{aharonov2007polynomial} as well). This restriction arises simply because $\SL{d}$ is not compact, and approximating elements very far from the identity requires longer sequences of gates.
For instance, it requires more applications of the gate
$\left(\begin{smallmatrix} 2 & 0 \\ 0 & 1/2 \end{smallmatrix}\right)$
to reach
$\left(\begin{smallmatrix} 2^{1000}& 0 \\ 0 & 2^{-1000} \end{smallmatrix}\right)$
than it requires to reach
$\left(\begin{smallmatrix} 2^{2} & 0 \\ 0 & 2^{-2} \end{smallmatrix}\right)$.
Since points arbitrarily far from the identity require arbitrarily long gate sequences to approximate, one cannot upper bound the length of sequences required to $\epsilon$-approximate arbitrary $M\in\SL{d}$ as a function of $\epsilon$ only---rather the length would depend on the distance of $M$ to the identity as well.
Restricting $M$'s distance to the identity allows one to upper bound the length of the approximating sequence in terms of $\epsilon$ only.

\end{document}